\documentclass[a4,twocolumn,showpacs,preprintnumbers,amsmath,amssymb,aps,pre]{revtex4}

\usepackage{graphics}
\usepackage{epsfig}
\usepackage{subfigure}

\usepackage{color, amsthm, mathrsfs}
\newtheorem*{thm*}{Theorem}
\newtheorem*{cor*}{Corollary}
\newtheorem*{lem*}{Lemma}
\let\Im\relax
\DeclareMathOperator{\Im}{Im}
\DeclareMathOperator{\Ker}{Ker}

\begin{document}

\title{Dimensionality reduction and spectral properties of multilayer networks}

\author{Rub\'en J. S\'anchez-Garc\'\i{}a}
\affiliation{Mathematical Sciences, University of Southampton, Highfield, Southampton SO17 1BJ, U.K.}
\email{R.Sanchez-Garcia@soton.ac.uk}

\author{Emanuele Cozzo}
\author{Yamir Moreno}
\affiliation{Institute for Biocomputation and Physics of Complex Systems (BIFI), Universidad de Zaragoza, 50018 Zaragoza, Spain}
\email{[yamir.moreno, emcozzo]@gmail.com}

\date{\today}

\begin{abstract}
Network representations are useful for describing the structure of a large variety of complex systems. Although most studies of real-world networks suppose that nodes are connected by only a single type of edge, most natural and engineered systems include multiple subsystems and layers of connectivity.  This new paradigm has attracted a great deal of attention and one fundamental challenge is to characterize multilayer networks both structurally and dynamically. One way to address this question is to study the spectral properties of such networks. Here, we apply the framework of graph quotients, which occurs naturally in this context, and the associated eigenvalue interlacing results, to the adjacency and Laplacian matrices of undirected multilayer networks. Specifically, we describe relationships between the eigenvalue spectra of multilayer networks and their two most natural quotients, the network of layers and the aggregate network, and show the dynamical implications of working with either of the two simplified representations. Our work thus contributes in particular to the study of dynamical processes whose critical properties are determined by the spectral properties of the underlying network.
\end{abstract}


\maketitle

\section{Introduction}\label{section:Intro}
Network theory has demonstrated to be an invaluable tool for studying complex system, i.e. systems composed of a large number of interacting elements. In particular, by analysing the spectral properties of the adjacency and Laplacian matrix of a network is possible to gain insight on the structure and dynamics occurring on the network \cite{vanmieghem2012complexspectra, brouwer2012spectra}. However, most natural and engineered complex systems occur in interaction with other complex systems and hence are better described by a multilayer network \cite{KivelaReview}. One can distinguish different types of multilayer networks depending on the interaction between the different systems (layers). 
For example, a \emph{multiplex network} is composed by elements that interact trough different channels. Each channel of interaction is represented by a layer, and the connections between different layers correspond to elements present in more than one layer simultaneously, so that in this case the intra-layer and inter-layer interactions represent indeed different kinds of relations.

Multilayer networks have attracted a lot of attention recently \cite{KivelaReview}, and many different structural and dynamical features of multilayer networks have been studied \cite{Barret2012sociality,DeDomenico2013tensorial,Baxter2012avalanche, Cardillo2013emergence,Cozzo2013clustering, Brummit2012cascades, Cozzo2012boolean}, demonstrating that the behavior of interacting complex systems is very different from a simple combination of the isolated cases. In this work we argue that the mathematical concept of \emph{quotient graph} (see Section \ref{section:NetworkQuotients} or \cite{haemers1995interlacing}) underpins the notion of multilayer network and gives fundamental insights into the structure and properties of the network, in particular its spectral properties. 

In the first part of this paper, we apply eigenvalue interlacing \cite{haemers1995interlacing} to the adjacency and Laplacian eigenvalues of multilayer network quotients and subnetworks. For the interlacing to hold for Laplacian eigenvalues, we define an appropriate notion of quotient Laplacian, and relate its eigenvalues to a Laplacian of the quotient graph. 
In the second part of the paper, we describe implications of the spectral results to the structure and dynamical processes on a multilayer network. In particular, we show how the pattern of connections between layers constraints the dynamics on the whole system. 
Our results agree with other methodologies such as perturbative analysis \cite{Sole-ribalta2013laplacianspectra,Radicchi2013abrupttransition} and put these and other results in a more rigorous framework.

A network quotient can be seen as a coarsening, reduction or simplification of the original network. In this sense our spectral results quantify the information loss about the eigenvalue spectrum resulting from this reduction process, expressed as certain eigenvalue inequalities. 

We define two natural quotients for a multilayer network: the network of layers, which represents the connection pattern between layers; and the aggregate network, which results from the projection of all layers onto an aggregated single-layer network (Fig.~\ref{fig2}). In addition, we consider each layer as a separate (sub)network. We then relate 
their adjacency and Laplacian eigenvalues to those of the whole multilayer network, as an interlacing result in the most general case, and as a lifting result is there is enough regularity in the connectivity patterns. We also consider the layer subnetworks, as their eigenvalues are related to the multilayer eigenvalues in a similar fashion. See Table \ref{table_summary} for a brief summary of the analytical results. The quotient point of view that we present also suggests a very concrete notion of aggregate network among the ones proposed in the literature \cite{Sole-ribalta2013laplacianspectra, Battiston2013metrics,DeDomenico2013tensorial}.

\begin{figure}
\includegraphics[width=\columnwidth]{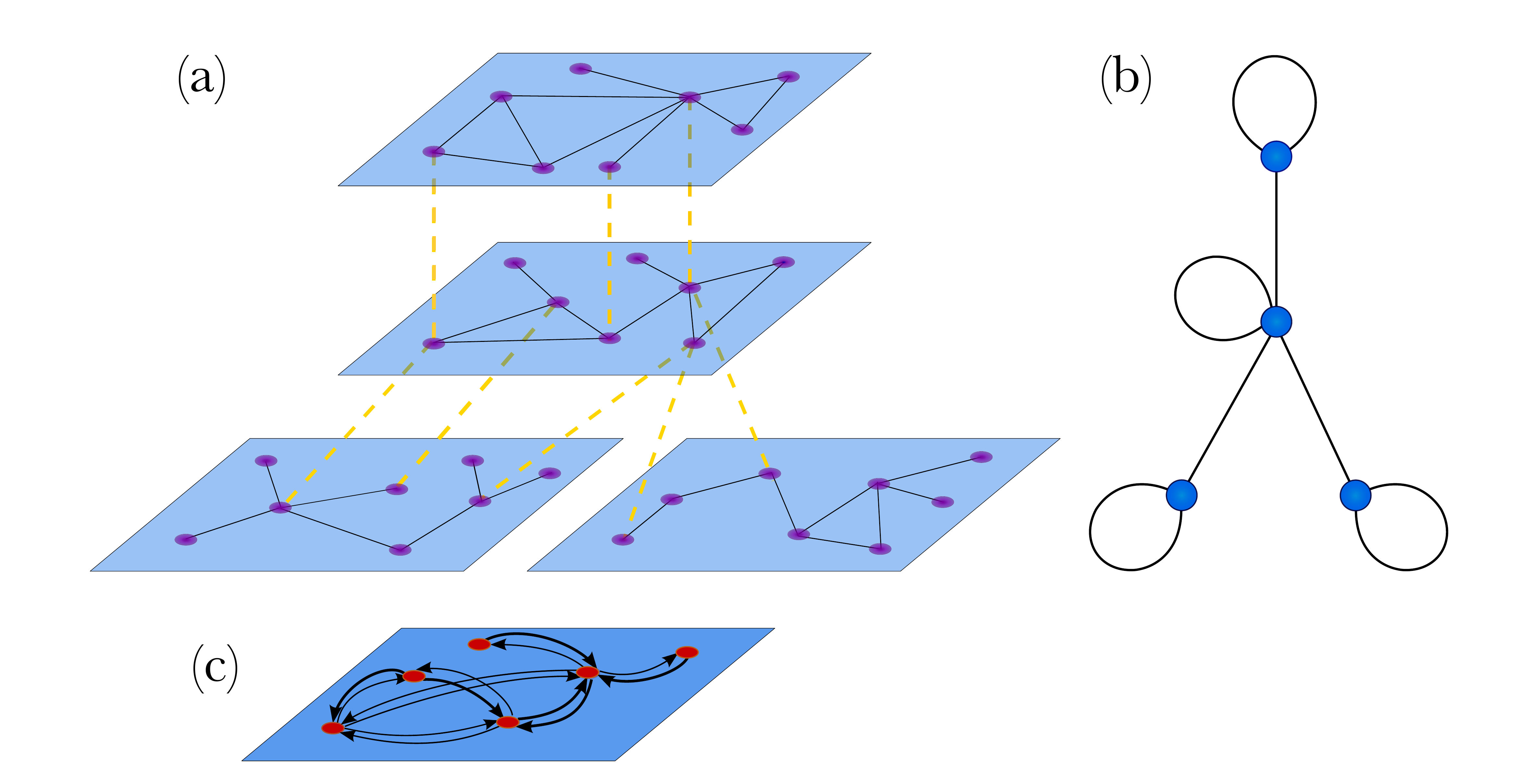}
\caption{(Color online) Schematic representation of a multilayer network with 4 layers and 8 nodes per layer (a), and its two quotients: the network of layers (b), and the aggregated network (c). In (a), dashed lines represent inter-layer edges. The quotient (b) is undirected, as all layers have the same number of nodes (see Eq.~\ref{eqn:NoL_weights}). The quotient (c) is only partially drawn, it is directed, and the edge thickness is proportional to the weight (Eq.~\ref{eqn:aggregate_weights}). 
The network of layers (b) corresponds to the layer interconnection structure, while the aggregate network (c) represents the superposition of all the layers onto one. In this sense, they can be thought of as `horizontal' and `vertical' quotients, as the figure suggests. Both quotients clearly represent a dimensionality reduction or coarsening of the original multilayer network.}
\label{fig2}
\end{figure}

\section{Mathematical background}\label{section:MathBackg}
We give a self-contained description of network quotients and interlacing results (deferring proofs to the Appendix), including regular quotients, and subnetworks. All the material presented here is well-known \cite{brouwer2012spectra}, except, as far as we know, the definition of quotient Laplacian (Eq.~\ref{eqn:QuotientLaplacian}) and its relation to the Laplacian of its quotient network.

\subsection{Adjacency and Laplacian matrices}
We represent an undirected network (or \emph{graph}) $\mathcal G$ on $n$ nodes by its adjacency matrix $A_{\mathcal G}=(a_{ij})$: $a_{ij}\neq 0$ represents an edge between nodes $i$ and $j$ with weight $a_{ij}$, while $a_{ij}=0$ if there is no such edge. Note that we allow positive and negative weights, and self loops ($a_{ii}\neq 0$). Any $n\times n$ real symmetric matrix is the adjacency matrix of such a network.

If the weights satisfy $a_{ij}=a_{ji}\ge 0$, we define the \emph{Laplacian matrix} as $L_{\mathcal G} = D-A_{\mathcal G}$, where $D=\text{diag}(d_1,\ldots,d_n)$ is the diagonal matrix of the node degrees
\begin{equation}
	d_i = \sum_{j=1}^n a_{ij} =\sum_{j=1}^n a_{ji}.
\end{equation}
(In this manuscript, by \emph{degree} we will always refer to weighted node degree as defined above.)

\subsection{Interlacing}
In this paper, we relate the adjacency and Laplacian eigenvalues of a multilayer network to two quotient networks that occur naturally. 
The main theoretical result that we will exploit is that the eigenvalues of a quotient interlace the eigenvalues of its parent network. Let $m<n$ and consider two sets of real numbers 
\[
	\mu_1 \le \ldots \le \mu_m \ \text{ and }\  \lambda_1 \le \ldots \le \lambda_n.
\] 
We say that the first set \emph{interlaces} the second if 
\[
	\lambda_i \le \mu_i \le \lambda_{i+(n-m)} \quad \text{ for } i=1, \ldots, m. 
\]

\subsection{Network quotients}\label{section:NetworkQuotients}
Suppose that $\{V_1, \ldots, V_m\}$ is a partition of the node set of a network $\mathcal{G}$ with adjacency matrix $A_\mathcal{G}$, and write $n_i=|V_i|$. The subnetwork represented by $V_i$ can be thought of as a cluster, community, or layer, for example. 

The \emph{quotient network} $\mathcal{Q}$ of $\mathcal{G}$ is a coarsening of the network with respect to the partition. It has one node per cluster $V_i$, and an edge from $V_i$ to $V_j$ weighted by an average connectivity from $V_i$ to $V_j$
\begin{equation}
	b_{ij} = \frac{1}{\sigma} \sum_{\substack{k \in V_i\\ l \in V_j}} a_{kl},
\end{equation}
where we have a choice for the size parameter $\sigma$: we will use either $\sigma_i=n_i$, or $\sigma_j=n_j$, or $\sigma_{ij}=\sqrt{n_i}\sqrt{n_j}$. We call the corresponding network the \emph{left quotient}, the \emph{right quotient} and the \emph{symmetric quotient} respectively. Fortunately, the matrix $B=(b_{ij})$ has the same eigenvalues for the three choices of $\sigma$ (see Appendix \ref{section:QuotientMatrix}). We refer by \emph{quotient network} to any of these three spectrally-equivalent networks with adjacency matrix $B$. Observe that the symmetric quotient is undirected, while the left and right quotients are not, unless all clusters have the same size, $n_i=n_j$ for all $i,j$.

The key spectral result is that the adjacency eigenvalues of a quotient network interlace the adjacency eigenvalues of the parent network (see Appendix \ref{appendix:interlacing} for a proof). The same result applies for Laplacian eigenvalues, if the Laplacian matrix of the quotient is defined appropriately, as explained below.

Consider the left quotient of $A$ with respect to the partition. Observe that the row sums of $Q_l(A)$ are
\begin{equation}
	\overline{d_i} =\frac{1}{n_i} \sum_{k \in V_i} d_k,
\end{equation}
the average node degree in $V_i$.  Let $\overline{D}$ be the diagonal matrix of the average node degrees. Then we define the \emph{quotient Laplacian} as the matrix 
\begin{equation}\label{eqn:QuotientLaplacian}
	L_\mathcal{Q} = \overline{D} - Q_l(A).
\end{equation}
(See Appendix \ref{appendix:Laplacian_quotients} for a full discussion on this choice.) With this definition, the Laplacian eigenvalues of the quotient network interlace the Laplacian eigenvalues of the parent network (see the Theorem in Appendix \ref{appendix:Laplacian_quotients}).

Let $\widetilde{\mathcal Q}$ be the \emph{loopless quotient} of $\mathcal G$, that is, the quotient network $\mathcal Q$ with all the self-loops removed. As the quotient Laplacian ignores self-loops (see Appendix \ref{appendix:Laplacian_quotients}), we have $L_{\mathcal Q}=L_{\widetilde{\mathcal Q}}$, and the interlacing result also holds for the loopless quotient.

%

\subsection{Regular quotients}\label{section:regularquotients}
A partition of the node set $\{V_1,\ldots,V_m \}$ is called \emph{equitable} if the number of edges (taking weights into account) from a node in $V_i$ to any node in $V_j$ is independent of the chosen node in $V_i$ 
\begin{equation}\label{eqn:regularity}
	\sum_{\substack{l \in V_j}} a_{kl} = \sum_{\substack{l \in V_j}} a_{k'l} \quad \text{ for all } k,k'\in V_i,
\end{equation}
for all $i,j$. This indicates a regularity condition on the connection pattern between (and within) clusters. If the partition is equitable, we call the quotient network \emph{regular}. A source of regular quotients are network symmetries 
\cite{MacArthur2008Symmetry,MacArthur2009Spectral}. For a toy example of a regular quotient, see Table \ref{table_adj}.

If the quotient is regular, the adjacency eigenvalues of $\mathcal Q$ not only interlace, but are a subset of the adjacency eigenvalues of $\mathcal G$ and, moreover, we can find an eigenbasis of $\mathcal G$ consisting on $m$ eigenvectors of the quotient \emph{lifted} to $\mathcal G$ (by repeating the coordinates on each cluster), and the other $n-m$ eigenvectors \emph{orthogonal} to the partition (the sum of the coordinates on each layer is zero); see Table \ref{table_adj} and Appendix \ref{appendix:equitable}. We refer to this spectral result as \emph{lifting}.

\begin{table}[!t]
\begin{center}
\begin{ruledtabular}
\begin{tabular}[c]{c  c  c}
Network & Eigenvalues & Eigenvectors \\
\hline 

\raisebox{-0.4cm}{\includegraphics[width=0.15\textwidth]{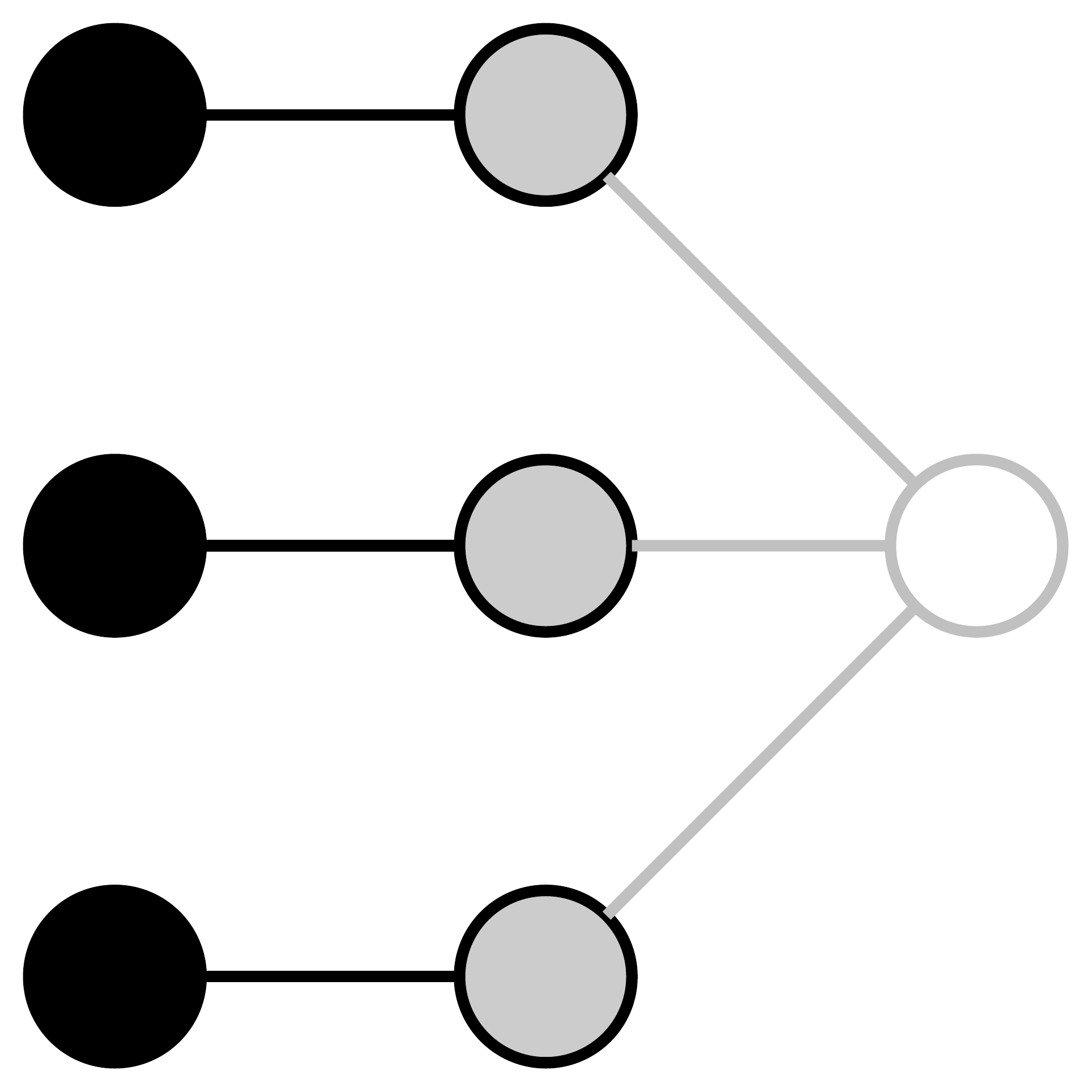}} & \begin{minipage}[b]{0.6cm} \footnotesize{\phantom{x} $1$\\ $1$ $-1$ $-1$ $2$ $-2$ $0$} \end{minipage} & \raisebox{-0.0cm}{\begin{minipage}[b]{4.0cm} \footnotesize{$(1, -1, 0, | 1, -1, 0, | 0)$ $(1, 0, -1, | 1, 0, -1, | 0)$ $(1, -1, 0, | -1, 1, 0, | 0)$ $(1, 0, -1, | -1, 0, 1, | 0)$ $(1, 1, 1, | 2, 2, 2, | 3)$ $(1, 1, 1, | -2, -2,  -2, | 3)$ $(1, 1, 1, | 0, 0, 0, | -1)$} \end{minipage}}\\

\hline 

\raisebox{0.17cm}{\includegraphics[scale=0.14]{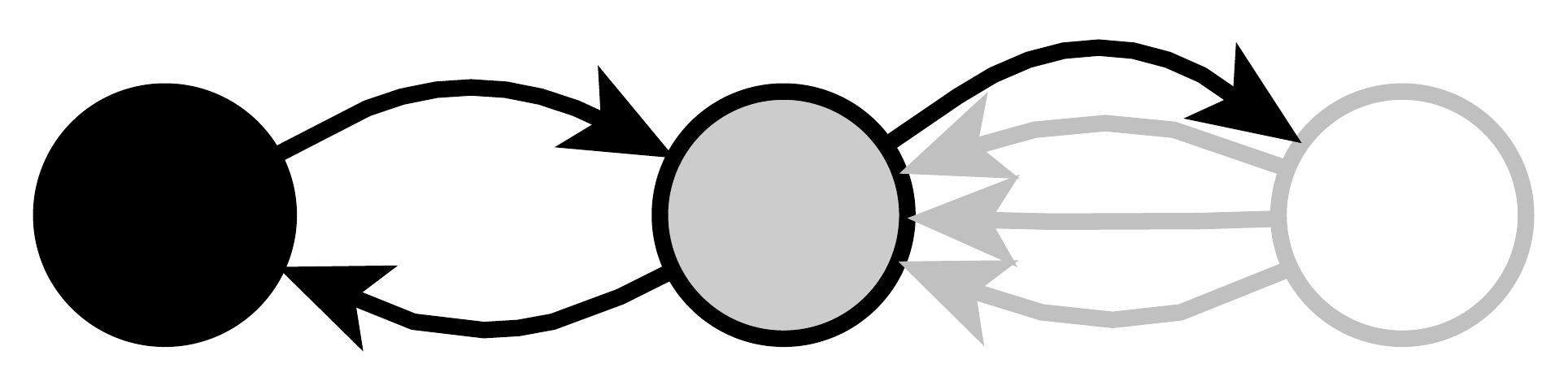}} & \raisebox{-0.0cm}{\begin{minipage}[b]{0.5cm} \footnotesize{\phantom{x} $2$ $-2$ $0$} \end{minipage}} & \raisebox{-0.0cm}{\begin{minipage}[b]{4.0cm} \footnotesize{$(1, 2, 3)$\\ $(1, -2, 3)$\\ $(1, 0, -1)$} \end{minipage}}\\

\end{tabular}
\end{ruledtabular}
\end{center}
\caption{\label{table_adj}\small{\textbf{Example of a regular quotient} (adapted from \cite{MacArthur2009Spectral}). We show the adjacency eigenvalues and a basis of eigenvectors for a simple network and a regular quotient. The colouring indicates the node set partition in three layers (represented vertically). Eigenvector entries on each layer are separated by vertical bars for convenience. Note that the spectrum of the quotient is a subset of the spectrum of the parent network. Moreover, the eigenbasis of the parent network consists of three eigenvectors of the quotient lifted to the parent graph (repeated coordinates on each layer) and the other eigenvectors are orthogonal to the partition (the sum of the coordinates on each layer is zero). The analogous result applies for the Laplacian eigenvalues, even if we add arbitrary intra-layer edges (almost regular quotient).}} 
\end{table}

For the Laplacian eigenvalues, the situation is somewhat simpler. We call a partition \emph{almost equitable} if condition (\ref{eqn:regularity}) is satisfied for all $i\neq j$ (but not necessarily for $i=j$), that is, if the regularity condition is satisfied after ignoring the intra-cluster edges.
In this case, we call the quotient graph $\mathcal Q$ \emph{almost regular}. Note that the quotient $\mathcal Q$ being almost regular is equivalent to the loopless quotient $\widetilde{\mathcal Q}$ being regular. 

The main result is that, if the quotient graph $\mathcal Q$ is almost regular, then the Laplacian eigenvalues of $\mathcal{Q}$ are a subset of the Laplacian eigenvalues of $\mathcal G$, and we can find a Laplacian eigenbasis of $\mathcal G$ consisting of $m$ Laplacian eigenvectors of the quotient ($\mathcal Q$ or $\widetilde{\mathcal Q}$) lifted to $\mathcal G$, and the other $n-m$ eigenvectors orthogonal to the partition (see Appendix \ref{appendix:Laplacian} for a proof). That is, we have a \emph{lifting} result for the Laplacian eigenvalues.

\subsection{Subnetworks}\label{section:Subnetworks}
Similar interlacing results apply when $B$ is a principal submatrix of $A$. If $A_\mathcal{G}$ is the adjacency matrix of a graph, a principal submatrix is the adjacency matrix of an induced subgraph. An \emph{induced subgraph} is a graph consisting on a subset of nodes and all the links between them. In contrast, a \emph{factor subgraph} consists on all the nodes and a subset of the links. A \emph{general subgraph} consists then of a subset of the nodes and a subset of the links between them.

For induced subgraphs, the adjacency eigenvalues of a induced subnetwork interlace the adjacency eigenvalues of the network (see Appendix \ref{appendix:interlacing}). For the Laplacian eigenvalues, only one of the interlacing inequalities hold, although this interlacing applies to general subgraphs, not necessarily induced. Namely, if $\lambda_1 \le \ldots \le \lambda_n$ are the Laplacian eigenvalues of a graph on $n$ vertices, and $\mu_1 \le \ldots \le \mu_m$ are the Laplacian eigenvalues of general subgraph on $m$ vertices, then 
\begin{equation}
	\mu_i \le \lambda_{i+(n-m)} \quad \text{for all } 1\le i \le m.
\end{equation} 
(See Appendix \ref{appendix:Laplacian_subnetworks} for a proof.)

These and the other spectral results are summarized on Table \ref{table_summary}.

\begin{table}[!t]
\begin{center}
\begin{ruledtabular}
\begin{tabular}{c|cc}
 & Adj.~eigenvalues & Laplacian eigenvalues\\
 \hline
 Quotient & \textit{interlacing} & \textit{interlacing}\\
 Almost reg.~quotient & \textit{interlacing} & \textit{lifting}\\
 Regular quotient & \textit{lifting} & \textit{lifting}\\
 \hline
 General subnetwork & --- & \textit{partial interlacing}\\
 Induced subnetwork & \textit{interlacing} & \textit{partial interlacing}
\end{tabular}
\end{ruledtabular}
\end{center}
\caption{\label{table_summary}\small{\textbf{Summary of spectral results}. This summarizes the spectral results in Section \ref{section:MathBackg}}, see main text for details.  
In the context of multilayer networks: The network of layers and aggregate network are examples of quotients; Regularity is a very strong condition for both quotients, but almost regularity may be satisfied by the network of layers is certain cases e.g. a layer-coupled multiplex \cite{KivelaReview}; the layer subnetworks are induced subnetworks.}
\end{table}

\section{Multilayer network quotients and spectra}\label{section:multilayerquotients}
Now we turn to exploit the spectral results on quotient networks in the framework of multilayer networks. After introducing the multilayer network formalism, we discuss two naturally occurring quotients: the network of layers and the aggregate network. We also discuss layer subnetworks, as similar interlacing results apply. This is not surprising, as subnetworks and quotients are dual concepts in some abstract categorical sense \cite{maclane1998categories}. 

Note that analogous results to those presented here will apply to arbitrary quotients or subnetworks on a multilayer network, and we only focus on the most natural ones. 
For the remainder, we implicitly assume the use of left quotients $Q(A)=Q_l(A)$ (cf.~Section \ref{section:NetworkQuotients}).

\subsection{Multilayer network formalism}\label{appendix:MultilayerFormalism}
We adopt the language and formalism of \cite{KivelaReview}. In most generality, a \emph{multilayer network} is a quadruplet $\mathcal M = (V_\mathcal{M}, E_\mathcal{M}, V, \mathbf{L})$ where $V$ is a set of nodes, $\mathbf{L}=\{L_a\}_{a=1}^d$ is a sequence of sets of layers, $V_\mathcal{M} \subseteq V \times \prod_{a=1}^dL_a$ are the multilayer network nodes (an element $(u,\alpha) \in V_\mathcal{M}$ represents node $u\in V$ in layer $\alpha$), and $E_\mathcal{M} \subseteq V_\mathcal{M} \times V_\mathcal{M}$ are the multilayer network edges. For simplicity, we assume from now on $d=1$, so there is only one set of layers $L$.

%

The pair $\mathcal G_{\mathcal{M}}=(V_\mathcal{M}, E_\mathcal{M})$ is a graph called the \emph{underlying graph} of the multilayer network. The \emph{supra-adjacency matrix} of $\mathcal{M}$ is the adjacency matrix of this graph. 
Besides, each layer can be considered as a subgraph $\mathcal{G}_\alpha =(V_\alpha, E_\alpha)$, where
\begin{eqnarray}
	V_\alpha &=&\{\left( u,\alpha \right) \in V_\mathcal{M}\}, \\
	E_\alpha &=& \{ \left( (u,\alpha), (v, \alpha) \right) \in E_\mathcal{M}\},
\end{eqnarray}
for each $\alpha \in L$. We write $A_\alpha$ for the adjacency matrix of $\mathcal{G}_\alpha$. The supra-adjacency matrix $A_{\mathcal M}=A_{\mathcal{G}_\mathcal{M}}$ has the matrices $A_\alpha$ as diagonal blocks, while the off-diagonal blocks $A_{\alpha\beta}$ represent inter-layer connectivity. 

Finally, we define the \emph{supra-Laplacian matrix} as the Laplacian of the underlying graph $L_{\mathcal M} = L_\mathcal{G_\mathcal{M}}$.

\subsection{Network of layers}\label{section:NoL}
The layers of a multilayer network partition the node set, so it is reasonable to consider the quotient induced by this partition. Let $\{V_1,\dots V_m\}$ be the partition of the multilayer node set by the layers, and $n_\alpha=\,|V_\alpha \,|$. 
Define the \emph{average inter-layer degree} from $\alpha$ to $\beta$ as
\begin{equation}\label{eqn:NoL_weights}
	d^{\alpha\beta} = \frac{1}{n_\alpha} \sum_{\substack{i \in V_\alpha\\ j\in V_\beta}} a_{ij}\,.
\end{equation}
This represents the average connectivity from a node in $\mathcal{G}_\alpha$ to any node in $\mathcal{G}_\beta$. 
If $\alpha=\beta$ we write $d^\alpha$ for $d^{\alpha\alpha}$, and call it the \emph{average intra-layer degree}.

Consider the quotient with respect to the partition given by the layers, that is, the (directed) network with adjacency matrix $(d^{\alpha\beta})$. We call this quotient the \emph{network of layers}.
Each node corresponds to a layer, with a self loop weighted by the average intra-layer degree $d^\alpha$, and there is a directed edge from layer $\alpha$ to layer $\beta$ weighted by the average inter-layer degree $d^{\alpha\beta}$.

Alternatively, we could consider the spectrally equivalent symmetric quotient, by replacing $1/n_\alpha$ by $1/(\sqrt{n_\alpha}\sqrt{n_\beta})$ in Eq.~\ref{eqn:NoL_weights}, see Section \ref{section:NetworkQuotients}. The network of layers will also be undirected if each layer contains the same number of nodes.

Applying the spectral results of Section \ref{section:NetworkQuotients}, we conclude that the adjacency, respectively Laplacian, eigenvalues of the network of layers interlace the adjacency, respectively Laplacian, eigenvalues of the multilayer network. Namely, if $\mu_1, \ldots, \mu_m$ are the (adjacency resp.~Laplacian) eigenvalues of the network of layers, then 
\begin{equation}
	\lambda_i \le \mu_i \le \lambda_{i+(n-m)} \quad \text{ for } i=1, \ldots, m, 
\end{equation}
where $\lambda_1, \ldots, \lambda_n$ are the (adjacency resp.~Laplacian) eigenvalues of the multilayer network.

The network of layers, ignoring weights and self-loops, simply represents the layer connection configuration (Fig.~\ref{fig2}). The connectivity of this reduced representation, measured in terms of the eigenvalues, thus relates to the connectivity of the entire multilayer network via the interlacing results. 

We turn to the question of when the layer partition is equitable. This requires, in particular, that the intra-layer degrees are constant, that is, each layer must be a $d^\alpha$-regular graph, a very strong condition unlikely to be satisfied in real-world multiplexes. Instead, we call a multilayer network \emph{regular} if the layer partition is almost equitable, that is, the inter-layer connections are independent of the chosen vertices. This is a more natural condition, and examples of inter-layer connections which give rise to regular multilayers are all-to-all, empty or one-to-one connections with homogeneous weights \cite{KivelaReview}. 

If the multilayer network is regular then, in addition to the interlacing, the Laplacian eigenvalues of the network of layers are a subset of the Laplacian eigenvalues of the multiplex, and we can lift a Laplacian eigenbasis of the quotient, as described in Section \ref{section:NetworkQuotients}. This latter result has also been derived in \cite{Sole-ribalta2013laplacianspectra}.

\subsection{Aggregate network}\label{section:aggregate}
The multilayer network formalism also includes information about nodes representing the same entity in several layers: given $u \in V$, we think of $(u,\alpha)$ and $(u,\beta)$ (if they are both multilayer nodes) as two nodes representing the same entity in two layers $\alpha \neq \beta$. This allows a second notion of quotient, the \emph{aggregate network}.

The aggregate network is obtained by identifying nodes representing the same `actor' or `component' in different layers (e.g.~same user in two social networks; same hub in different transport networks; multiplexes describing time series \cite{KivelaReview}). This identification also makes sense for interdependent networks where the functioning of a node in a layer critically depends on the functioning of another node in another layer and vice versa \cite{Vespignani2010interdependent}. Several candidates for this aggregate network  have been proposed in the literature such as the average network \cite{Sole-ribalta2013laplacianspectra}, the overlapping network \cite{Battiston2013metrics} the projected monoplex network \cite{DeDomenico2013tensorial} or the overlay network \cite{DeDomenico2013tensorial}. We claim that the natural definition of an aggregate network is given by the suitable notion of quotient network, as follows. 

We define a \emph{supra-node} as the set of nodes representing the same object 
\begin{equation}
	\widetilde{u} = \{ (u, \alpha) \in V_\mathcal{M} \,|\, \alpha \in L\}.
\end{equation}
Note that not every node is present in every layer, and $\widetilde{u}$ may have cardinality 1. We call $\kappa_{\widetilde{u}}=|\widetilde{u}|$ the \emph{multiplexity degree} of the supra-node $\widetilde{u}$, that is, the number of layers in which an instance of the same object $u$ appears. We also define the \emph{average connectivity} between supra-nodes $\widetilde{u}$ and $\widetilde{v}$ as
\begin{equation}\label{eqn:aggregate_weights}
	d_{\widetilde{u}\widetilde{v}} = \frac{1}{k_{\widetilde{u}}} \,\sum_{\substack{i \in \widetilde{u}\\ j\in\widetilde{v}}} a_{ij},
\end{equation}
and write $d_{\widetilde{u}}$ for $d_{\widetilde{u}\widetilde{u}}$.
   

Observe that the super-nodes partition the multilayer node set. We define the \emph{aggregate network} as the quotient associated with this partition. Each node in this quotient corresponds to a supra-node, with a self-loop weighted by $d_{\widetilde{u}}$, and a directed edge from $\widetilde{u}$ to $\widetilde{v}$ weighted by $d_{\widetilde{u}\widetilde{v}}$.

Alternatively, we could consider the symmetric quotient, which is an undirected network and has the same eigenvalues, by simply replacing $1/k_{\widetilde{u}}$ by $1/(\sqrt{\kappa_{\widetilde{u}}}\sqrt{\kappa_{\widetilde{v}}})$ in Eq.~\ref{eqn:aggregate_weights}. Note that the aggregate network quotient will also be undirected if every supra-node has the same multiplexity degree (cf.~Section \ref{section:NetworkQuotients}).

Finally, using the spectral results of Section \ref{section:NetworkQuotients}, we conclude that the adjacency (respectively Laplacian) eigenvalues of the aggregate network interlace the adjacency (respectively Laplacian) eigenvalues of the multiplex. Namely, in a multilayer network with $n$ nodes and $\widetilde{n}$ supra-nodes, the (adjacency resp.~Laplacian) eigenvalues of the aggregate network quotient $\mu_1,\ldots,\mu_{\widetilde{n}}$ satisfy
\begin{equation}
	\lambda_i \le \mu_i \le \lambda_{i+(n-\widetilde{n})} \quad \text{ for } i=1, \ldots, \widetilde{n}, 
\end{equation}
where $\lambda_1, \ldots, \lambda_n$ are the (adjacency resp.~Laplacian) eigenvalues of the multilayer network.

Observe that requiring the aggregate network to be regular, or almost regular, is in this case very restrictive, as it would require that every pair of nodes connects in the same uniform way on every layer, and thus it is not likely to occur on real-world multilayer networks.

\subsection{Layer subnetworks}\label{layerspectra}
The layers of a multiplex form evident subnetworks, and it is natural to relate the eigenvalues of each layer to the eigenvalues of the multiplex. As we have seen (Section \ref{section:Subnetworks}), the interlacing result applies to the adjacency eigenvalues of an induced subnetwork, such as the layers, and partial interlacing also holds for the Laplacian eigenvalues. More precisely, if a layer subgraph $\mathcal{G}_\alpha$ has $n_\alpha$ nodes and adjacency (resp.~Laplacian) eigenvalues $\mu_1,\ldots, \mu_\alpha$, and 
$\lambda_1, \ldots, \lambda_n$ are the adjacency (resp.~Laplacian) eigenvalues of the whole multilayer network, then
 \begin{eqnarray}
	\lambda_i \le \mu_i \le \lambda_{i+(n-n_\alpha)} && \text{ for } i=1, \ldots, n_\alpha, \ \text{ resp.}\\
	\mu_i \le \lambda_{i+(n-n_\alpha)} && \text{ for } i=1, \ldots, n_\alpha.
\end{eqnarray}

\section{Discussion and Applications}\label{section:Discussion}
From a physical point of view, the adjacency and Laplacian spectra of a network encode information on structural and dynamical properties of the system represented by the network. We now discuss some consequences and applications of the spectral results derived in the previous sections. In the following, let us write $\lambda_i(A)$ for the $i$th smallest eigenvalue of a matrix $A$.

\subsection{Adjacency spectrum}
The spectrum of the adjacency matrix is directly related to different dynamical processes that take place on the system, such as spreading processes, for which it has been shown that critical properties are related to the inverse of the largest eigenvalue of this matrix. As an example, consider a contact process on the multilayer network $\mathcal{M}$ whose dynamic is described by the equation
\begin{equation}
p_i(t+1)=\beta\sum_ja_{ij}p_j(t)-\mu\, p_i(t)
\label{contact}
\end{equation} 
in which $p_i(t)$ is the probability of node $i$ to be infected at time $t$, $\beta$ is the infection rate, $\mu$ is the recovery rate and $a_{ij}$ are the elements of the supra-adjacency matrix $A_{\mathcal{M}}$. In this model, each infected node contacts its neighbours with probability $1$, and tries to infect them. The contact between two instances of the same object in different layers is modelled in the same way as the contact between any two other nodes (the layer structure is ignored).
The critical value for which the infection survives is given by
\begin{equation}
\beta_c=\frac{\mu}{\lambda_n(A_\mathcal{M})}.
\label{betacritic}
\end{equation}

From the interlacing result for the layer subnetworks (Section \ref{layerspectra}) we have that
\begin{equation}
\lambda_{n_\alpha}(A_\alpha)\leq \lambda_n (A_\mathcal{M}),
\end{equation}
where $A_\alpha$ is the adjacency matrix of the layer $\alpha$. This means that the critical point for the multilayer network $\beta_c$ is bounded from above by the corresponding critical points of the independent layers \cite{cozzo:2013}. This implies that the multilayer network is more efficient as far as a spreading processes are concerned than the most efficient of its layers on its own.

On the other hand, if $\lambda_m$ is the largest adjacency eigenvalue of the network of layers, then (Section \ref{section:NoL})
\begin{equation}
\lambda_m\leq \lambda_n(A),
\end{equation}
which means that the connections between layers also impose constraints to the dynamics on the multilayer network. In particular, the critical point of the spreading dynamics on the multilayer network is bounded from above by the corresponding critical point of the network of layers. Note that this also explains the existence of a mixed phase \cite{Dickinson2010epidemicinterconnected}.

Consider now the same process (\ref{contact}), this time defined on the aggregate network
\begin{equation}
p_{\tilde{u}}(t+1)=\beta\sum_{\tilde{v}} a_{\tilde{u}\tilde{v}}p_{\tilde{v}}(t)-\mu \, p_{\tilde{u}}(t).
\label{agcontact}
\end{equation}
Here $a_{\tilde{u}\tilde{v}}$ are the elements of $Q(A_{\mathcal{M}})$, the adjacency matrix of the aggregate graph. The critical value is given by
\begin{equation}
\widetilde{\beta}_{c}=\frac{\mu}{\lambda_{\tilde{n}}(Q(A_\mathcal{M}))}
\label{agbetacritic}
\end{equation}
where $\tilde{n}$ is the number of supra-nodes in $\mathcal{M}$ (the size of the aggregate network). From the interlacing result we have that 
\[
\widetilde{\beta}_{c}\ge\beta_c.
\]
Therefore the spreading process on $\mathcal{M}$ is at least as efficient as the same spreading process on the aggregate network. 

Note that Equations \ref{contact} and \ref{agcontact} describe two rather different processes, that is, two different strategies that actors can adopt in order to spread information across the multilayer network. In the former, a node can infect any other node on any layer, while in the latter, each supra-node chooses at each time step with uniform probability a layer in which an instance representing it is present and then contacts all its neighbours in that layer. Our results show that the former strategy is more effective than the latter, as expressed by the relation between the critical points.

\subsection{Laplacian spectrum}
The Laplacian of a network $L=(l_{ij})$ is the operator of the dynamical process described by
\begin{equation}
\dot{p}_{ij}(t)=-\sum_{k}p_{ik}(t)\,l_{ki}
\label{multidiff}
\end{equation}
where $p_{ij}(t)$ represents the transition probability of a particle from node $i$ to node $j$ at time $t$. 
The second smallest eigenvalue of the Laplacian matrix sets the time scale of the process.  From the interlacing results applied to the Laplacian matrix we have that for any quotient
\begin{equation}\label{layernet}
\lambda_2(L_\mathcal{M})\leq \lambda_2(Q(L_\mathcal{M})).
\end{equation}
That is, the relaxation time on the multiplex is at most the relaxation time on any quotient, in particular the network of layers or the aggregate network. If we interpret $\lambda_2$ of a network Laplacian as algebraic connectivity \cite{brouwer2012spectra}, Eq.~\ref{layernet} means that the algebraic connectivity of the multilayer network is always bounded above by the algebraic connectivity of any of its quotients. 

As a more concrete example of the above, consider a multilayer network describing a time series. Then the network of layers is a path graph on $m$ nodes (the number of layers) and hence 
\begin{equation}
\lambda_2(L_\mathcal{M})\leq 2-2\cos\left(\frac{\pi}{m}\right).
\end{equation}
This means that in this case the relaxation time is proportional to the length of the time series, as one would expect.

On the other hand, the Laplacian of the aggregated network is the operator corresponding to the dynamical process described by
\begin{equation}
\dot{p}_{\tilde{u}\tilde{v}}(t)=\sum_{\tilde{k}} p_{\tilde{u}\tilde{k}}(t) \,a_{\tilde{k}\tilde{v}}-d_{\tilde{u}}\,
p_{\tilde{u}\tilde{v}}
(t)=\sum_{\tilde{k}}p_{\tilde{u}\tilde{k}}(t)\,\tilde{l}_{\tilde{k}\tilde{u}}
\label{aggregatediff}
\end{equation}
where $p_{\tilde{i}\tilde{j}}(t)$ is the transition probability of a particle from supra-node $\tilde{u}$ to supra-node $\tilde{v}$ at time $t$, $a_{\tilde{u}\tilde{k}}$ are the elements of the adjacency matrix of the aggregated contact network, $\tilde{L}=(\tilde{l}_{ij})$ is the Laplacian matrix of the aggregate contact network (i.e.~$\tilde{L}=Q(L_\mathcal{M})$) and $d_{\tilde{u}}=\sum_{\tilde{v}}a_{\tilde{u}\tilde{v}}$ is the degree  of a supra-node. 
Note that if we define the overlapping degree \cite{Battiston2013metrics} of a supra-node as
\[
o_{\tilde{u}}=\sum_{\tilde{v}} a_{\tilde{u}\tilde{v}}
\]
then we have that
\[
d_{\tilde{u}}=\frac{1}{\kappa_{\tilde{u}}}o_{\tilde{u}}.
\]
From the interlacing result for the Laplacian we have that
\begin{equation}
\lambda_2(L_\mathcal{M})\leq \lambda_2(Q(L_\mathcal{M})).
\label{final}
\end{equation}
That is, the diffusion process on the aggregate network (Eq.~\ref{aggregatediff}) is faster than the diffusion process on the entire multilayer network (Eq.~\ref{multidiff}). 

Note that in \cite{Sole-ribalta2013laplacianspectra}, in a setting in which all nodes are connected to a counterpart in each layer, the authors obtained by means of a perturbative analysis that $\lambda_2(L_\mathcal{M})\sim \lambda_2(Q(L_\mathcal{M}))$ when the diffusion parameter between layers is large enough. In \cite{Radicchi2013abrupttransition} this result is generalized (in a different framework, since they are interested in structural properties of interdependent networks) to all almost regular multilayer networks. In the framework of quotient networks that we have presented here those results arise in a very natural way. Besides, eigenvalue interlacing between multilayer and quotient eigenvalues holds for every possible inter-layer connection scheme.   

In the context of synchronization, the smallest non-zero Laplacian eigenvalue $\lambda_2$ is also related to the stability of a synchronized state  \cite{almendral-diaz2007dynspec}, and indeed the larger $\lambda_2$ is, the more stable is the synchronized state. Considering a multilayer network, the bound in (\ref{layernet}) means that the synchronized state of a system supported on the multilayer network is at most as stable as the synchronized state on any of its quotients.

\section{Conclusions} 
We have presented the network quotient formalism in the context of multilayer networks, highlighting the two most natural quotients, the network of layers, and the aggregate network. Structurally, a quotient can be thought as a dimensionality reduction of a multilayer network. In terms of spectra, we have showed that eigenvalue interlacing applies to the adjacency and Laplacian eigenvalues of any quotient, and also to subnetworks such as the layer subnetwork. We needed in particular a definition of quotient Laplacian, and to relate its eigenvalues to those of a Laplacian of the quotient network. 
We have also investigated regularity of the inter-layer connectivity, which gives a stronger lifting result on the eigenvalues and eigenvectors. Finally, we have discussed possible applications of our results, including reproducing previous results in the literature obtained by other means such as perturbative analysis. 

We argue that the notion of quotient is closely intertwined to that of multilayer network, as the latter formally corresponds to an ordinary network with additional layer and node identification information. Thinking of a network quotient as a partition or identification of its node set, a multilayer network can be indeed recovered from its underlying network and these two quotients, the network of layers, and the aggregate network. We hope that the quotient point of view will be a useful and complementary perspective in the study of multilayer networks.

\begin{acknowledgements}
We thank M.~A.~Porter and B.~MacArthur for providing useful comments on an earlier draft. E.~C.~was supported by the FPI program of the Government of Arag\'on, Spain. This work has been partially supported by the EPSRC grant EP/G059101/1 (U.K.); the MINECO grant FIS2011-25167 (Spain); Comunidad de Arag\'on (Spain) through a grant to the group FENOL, and by the EC FET-Proactive Project PLEXMATH (grant 317614).
\end{acknowledgements}

\appendix
\section{Mathematical statements}
\subsection{The quotient of a symmetric matrix}\label{section:QuotientMatrix}
The quotient formalism holds in more generality for any real symmetric matrix, as we explain here.
Let $A=(a_{ij})$ be any real symmetric $n \times n$ matrix. Write $X=\{1,2,\ldots,n\}$, let $\{X_1,\ldots,X_m\}$ be a partition of $X$, and let $n_i=|X_i|$. We write $A_{ij}$ for the submatrix consisting of the intersection of the $k$-rows and $l$-columns of $A$ such that $k \in X_i$ and $l \in X_j$. In particular, $A_{ij}$ is an $n_i \times n_j$ matrix. Define $b_{ij}$ as the average row sum of $A_{ij}$,
\begin{equation}\label{eqn:defbij}
	b_{ij}=\frac{1}{n_i} \sum_{\substack{k \in X_i\\ l \in X_j}} a_{kl}.
\end{equation}
The $m \times m$ matrix $Q_l(A)=(b_{ij})$ is called the \emph{left quotient matrix of $A$} with respect to the partition $\{X_1,\ldots,X_m\}$.

We can express $Q_l(A)$ in matrix form, as follows. Let $S=(s_{ij})$ be the $n \times m$ \emph{characteristic matrix} of the partition, that is, $s_{ij}=1$ if $i \in X_j$, and 0 otherwise. Then $S^TAS$ is the matrix of coefficient sums of the submatrices $A_{ij}$, and, hence, $Q_l(A)=\Lambda^{-1}S^TAS$, where $\Lambda=\text{diag}(n_1,\ldots,n_m)$. 

There are two alternatives to $Q_l(A)$, called the \emph{right quotient} and the \emph{symmetric quotient}, written $Q_r(A)$ and $Q_s(A)$. They correspond to replacing $1/n_i$ in (\ref{eqn:defbij}) by $1/n_j$ respectively $1/\sqrt{n_i}\sqrt{n_j}$. In matrix form, we have $Q_r(A)=S^TAS\Lambda^{-1}$ and $Q_s(A)=\Lambda^{-1/2}S^TAS\Lambda^{-1/2}$. Note that $Q_l(A)$ is the transpose of $Q_r(A)$, and they are not symmetric unless $n_i=n_j$ for all $i,j$.

Nevertheless, these three matrices have the same spectrum (the proof is straightforward):
\begin{lem*}
Let $X, D$ be $m\times m$ matrices, with $D$ diagonal. Then the matrices $DX$, $XD$ and $D^{1/2}XD^{1/2}$ have all the same spectrum. 
\end{lem*}


The key result is that the eigenvalues of a quotient matrix interlace the eigenvalues of $A$, as we explain next. From now on let $Q(A)=Q_l(A)$, the quotient matrix normally referred to in the literature.

\subsection{Interlacing eigenvalues}\label{appendix:interlacing}
All the interlacing results we refer to are a consequence of the theorem below, which in turn follows from the Courant-Fisher max-min theorem. 
\begin{thm*}[\protect{\cite[Thm.~2.1(i)]{haemers1995interlacing}}] Let $A$ be a symmetric matrix of order $n$, and let $U$ be an $n\times m$ matrix such that $U^TU=I$. Then the eigenvalues of $U^TAU$ interlace those of $A$.
\end{thm*}

Observe that the matrix $U^TAU$ is symmetric, and hence it has real eigenvalues.

If $U$ is the characteristic matrix of a subset $\alpha \subset \{1,2,\ldots,n\}$, that is, $U=(u_{ij})$ of size $n \times |\alpha|$ and non-zero entries $u_{ii}=1$ if $i \in \alpha$, then $U^TAU$ equals the principal submatrix of $A$ with respect to $\alpha$. As $U^TU$ is the identity, we conclude from the theorem above:

\begin{cor*}[\protect{\cite[Cor.~2.2]{haemers1995interlacing}}]
Let $B$ be a principal submatrix of $A$. Then the eigenvalues of $B$ interlace the eigenvalues of $A$.
\end{cor*}

On the other hand, if $S$ is the characteristic matrix of the partition, then $S^TS=\Lambda$ is a diagonal non-singular matrix, and hence $U=S\Lambda^{-1/2}$ satisfies the hypothesis of the theorem. We conclude that the eigenvalues of $U^TAU=\Lambda^{-1/2}S^TAS\Lambda^{-1/2}$ interlace those of $A$. Using the Lemma in \ref{section:QuotientMatrix}, we conclude:
\begin{cor*}[\protect{\cite[Cor.~2.3(i)]{haemers1995interlacing}}]
Let $B$ be a quotient matrix of $A$ with respect to some partition. Then the eigenvalues of $B$ interlace the eigenvalues of $A$.
\end{cor*}

\subsection{Equitable partitions}\label{appendix:equitable}
A partition of the node set is called \emph{equitable} if, for each $i, j$, the row sum of the submatrix $A_{ij}$ is constant, that is, 
\begin{equation}\label{eqn:equitable}
	\sum_{l \in X_j} a_{kl} = \sum_{l \in X_j} a_{k'l} \quad \text{for all } k,k'\in X_i.
\end{equation}	
This can be expressed in matrix form as $A\,S=S\,Q(A)$. We call the matrix $Q(A)$ a \emph{regular quotient} if it is the quotient of an equitable partition.

If the quotient is regular, then the eigenvalues of $Q(A)$ not only interlace but are a subset of the eigenvalues of $A$. In fact, there is a \emph{lifting} relating both sets of eigenvalues, as we explain now. 

If $v, w$ are column vectors of size $m$ and $n$, we say that $Sv$ represents the vector $v$ \emph{lifted} to $A$, and $S^Tw$ the vector $w$ \emph{projected} to $Q(A)$. The vector $Sv$ has constant coordinates on each $X_i$, while the vector $S^Tw$ is created by adding the coordinates on each $X_i$. The vector $w$ is called \emph{orthogonal to the partition} if $S^Tw=0$, that is, the sum of the coordinates over each $X_i$ is zero.

If the quotient is regular, the spectrum of $A$ decomposes into the spectrum of $B$ lifted to $A$ (i.e.~eigenvectors constant on each $X_i$), and the remaining spectrum is orthogonal to the partition (i.e.~eigenvectors with coordinates adding to zero on each $X_i$):

\begin{thm*}
Let $B$ be the quotient matrix of $A$ with respect to an equitable partition with characteristic matrix $S$. Then the spectrum of $B$ is a subset of the spectrum of $A$. More precisely, $(\lambda, v)$ is an eigenpair of $B$ if and only if $(\lambda, Sv)$ is an eigenpair of $A$.\\
Moreover, there is an eigenbasis of $A$ of the form $\{Sv_1,\ldots, Sv_m,w_1,\ldots, w_{n-m}\}$ such that $\{v_1, \ldots, v_m\}$ is any eigenbasis of $B$, and $S^Tw_i=0$ for all $i$.
\end{thm*}

\begin{proof}The first part follows easily from the identity $SA=SB$, which is equivalent to Eq.~\ref{eqn:equitable} (note that $Sv\neq 0$ as $\Ker(S)=0$). For the second part, note that $S$ is an isomorphism onto $\Im(S)$, as it has trivial kernel, so $\{Sv_1,\ldots,Sv_m\}$ is a basis of $\Im(S)$. It is easy to show that the orthogonal complement $\Im(S)^\perp$ equals $\Ker(S^T)$, hence we can complete the linearly independent set of eigenvectors $\{Sv_1,\ldots,Sv_m\}$ to a eigenbasis of $\mathbb{R}^n=\Im(S) \oplus \Im(S)^\perp$.
\end{proof}


\subsection{Laplacian eigenvalues} \label{appendix:Laplacian}
\subsubsection{Quotients}\label{appendix:Laplacian_quotients}
We want to show that the Laplacian of a quotient graph is the quotient of the Laplacian matrix, as this will allow us to extend the interlacing results to the Laplacian eigenvalues. First, we need to clarify what we mean by the Laplacian of a non-symmetric matrix.

If $A=(a_{ij})$ is a real symmetric (adjacency) matrix, define the \emph{node out-degrees} as
\begin{equation}
	d^{out}_i=\sum_j a_{ij}	\quad \text{(row sum)}.
\end{equation}
The \emph{out-degree Laplacian} is the matrix 
\begin{equation}
	L^{out}=D^{out}-A,
\end{equation}
where $D^{out}$ is the diagonal matrix of the out-degrees. We define $d_i^{in}$, $D^{in}$ and the \emph{in-degree Laplacian} $L^{in}$ analogously. Note that both Laplacian matrices ignore the diagonal values of $A$. (If $A$ is the adjacency matrix of a graph, we say that the Laplacian `ignores self-loops'.)

Consider the left and right quotients of $A$ with respect to a given partition. Observe that the row sums of $Q_l(A)$ are
\begin{equation}
	\overline{d_i} =\frac{1}{n_i} \sum_{k \in V_i} d_k,
\end{equation}
the average node degree in $V_i$.  Let $\overline{D}$ be the diagonal matrix of the average node degrees. Then we define the \emph{quotient Laplacian} as the matrix 
\begin{equation}
	L_\mathcal{Q} = \overline{D} - Q_l(A),
\end{equation}
that is, the out-degree Laplacian of the left quotient matrix. Alternatively, we could have defined $L_\mathcal{Q}$ as the in-degree Laplacian of the right quotient matrix, giving a transpose matrix with the same eigenvalues. (Note that there is no obvious way of interpreting the symmetric quotient $Q_s(L)$ as the Laplacian of a graph.) 

Now we can prove that the Laplacian of the quotient is the quotient of the Laplacian, in the following sense.
\begin{thm*}
Let $\mathcal G$ be a graph with adjacency matrix $A$ and Laplacian matrix $L$. Then:
\[
	L^{out}(Q_l(A))=Q_l(L).
\]
The analogous result holds for the right quotients and the in-degree Laplacian.
\end{thm*}
\begin{proof}
By definition (see Appendix \ref{section:QuotientMatrix}), 
\begin{align*}
	Q_l(L) &= \Lambda^{-1}S^TLS = \Lambda^{-1}S^T(D-A)S = \\
	&= \Lambda^{-1}S^TDS - \Lambda^{-1}S^TAS = \overline{D} - Q_l(A). 
\end{align*}
The second statement follows by transposing the equation above.
\end{proof}
This theorem allows us to use the interlacing results of Appendix \ref{appendix:interlacing} for Laplacian eigenvalues.

We finish by studying equitable partitions in the context of Laplacian matrices. We demonstrate that a partition being regular for the Laplacian matrix is equivalent to the partition being almost regular for the adjacency matrix. In particular, the spectral results of Appendix \ref{appendix:equitable} will hold for almost regular quotients and Laplacian eigenvalues. 

\begin{thm*}
Let $\mathcal G$ be a graph with adjacency matrix $A$ and Laplacian matrix $L$. Then a partition is equitable with respect to $L$ if and only if it is almost equitable with respect to $A$.
\end{thm*}
\begin{proof}
By relabeling the nodes if necessary, we can assume the block decomposition
\begin{equation}
	A=
	\begin{pmatrix}
		A_{11} &\ldots &A_{1m}\\
		\vdots & \ddots &\vdots\\
		A_{m1} &\ldots & A_{mm}
	\end{pmatrix},	
	\label{subnetmat}
\end{equation}
where the $n_i \times n_j$ submatrix $A_{ij}$ represents the edges from $V_i$ to $V_j$. The matrix $L$ has then a similar block decomposition into submatrices $L_{ij}$. As $L=D-A$ and $D$ is diagonal, we have $L_{ij}=-A_{ij}$ for all $i\neq j$. In particular, the row sums of $L_{ij}$ are constant if and only if the row sums of $A_{ij}$ is constant, for all $i \neq j$. On the other hand, as the row sums in $L$ are zero, the row sums in $L_{ii}$ equals the sum of the row sums of the matrices $L_{ij}$ for $j\neq i$, and the result follows.
\end{proof}

\subsubsection{Subnetworks}\label{appendix:Laplacian_subnetworks}
The adjacency eigenvalues of an induced subgraph interlace those of the graph, as per the first corollary in Appendix \ref{appendix:interlacing}. However, they do not behave well for factor subgraphs \cite{haemers1995interlacing}. For the Laplacian eigenvalues, they behave well for general subgraphs, in the following sense.

\begin{thm*}
Let $\lambda_1 \le \ldots \le \lambda_n$ be the Laplacian eigenvalues of a graph $\mathcal G$ on $n$ vertices, and let $\mu_1 \le \ldots \le\mu_m$ be the Laplacian eigenvalues of a general subgraph on $m$ vertices. Then $\mu_i \le \lambda_{i+(n-m)}$ for all $1\le i\le m$.
\end{thm*}

\begin{proof}The proof uses the same ideas described in Appendix \ref{appendix:interlacing} (cf.~\cite[Prop.~3.2.1(ii)]{brouwer2012spectra}). Write $L=BB^T$ where $B$ is the incidence matrix of the graph. The matrices $BB^T$ and $B^TB$ have the same non-zero spectrum (probably with different multiplicities), and removing an edge is equivalent to taking a principal submatrix of $B^TB$, hence interlacing applies. To remove a vertex, first remove all the incident edges, then removing the vertex corresponds to removing a zero eigenvalue. \end{proof}

\end{document}